\documentclass{birkjour}

\usepackage[english]{babel}

 \newtheorem{thm}{Theorem}[section]

 \newtheorem{prop}[thm]{Proposition}
 \theoremstyle{definition}
 
 \theoremstyle{remark}

 \numberwithin{equation}{section}
\usepackage{amsmath}
\usepackage{graphicx}
\usepackage[colorlinks=true, allcolors=blue]{hyperref}

\title [Scattering matrices for perturbations of Laplace operator]{Scattering matrices for perturbations of Laplace operator by infinite sums of zero-range potentials}
\author[V. Adamyan]{Vadym Adamyan } 
\address{
Odessa I.I. Mechnikov National University\\
65082 Odessa\\ 
Ukraine}

\email{vadamyan@onu.edu.ua}

\thanks{}
\begin{document}

\begin{abstract}
This paper analyzes the scattering matrix for two unbounded self-adjoint operators: the
standard Laplace operator in three-dimensional space and a second operator
that differs from the first by an infinite sum of zero-range potentials.
\end{abstract}

\maketitle

\section{Introduction}

Let $A,A_{1}$ be unbounded selfadjoint operators in Hilbert space $\mathcal{H}$, $\mathcal{D},\mathcal{D}_{1}$ and $R(z),R_{1}(z), \, \mathrm{Im}z\neq 0,$ are the domains and resolvents of $A,A_{1}$, respectively. Following \cite{AK} we call $A_{1}$ \textit{singular perturbation} of $A$ if
 $$ \begin{array}{cc} 
1)\,  \mathcal{D}\cap\mathcal{D}_{1} \, \text{is dense in }  \mathcal{H}; &
2)\, A_{1}=A \, \text{on } \mathcal{D} \cap \mathcal{D}_{1}. \end{array} 	
$$ 
In accordance with this definition, $A,A_{1}$ are selfadjoint extensions of the same densely defined symmetric operator
$$ B= A|_{\mathcal{D}\cap\mathcal{D}_{1}}=A_{1}|_{\mathcal{D}\cap\mathcal{D}_{1}} . $$
Therefore, the specific properties of the operator $A_{1}$ as a perturbation of the given operator $A$ can be investigated in the framework of the extension theory of symmetric operators.  
Although this path ultimately leads to the desired results, it is, as a rule, quite complicated and confusing. At that, the range of significant problems related to real applications that can be posed and solved in the framework of the theory of singular perturbations of selfadjoint operators is quite extensive (see \cite{AGHH}, \cite{EKKST}). This is because, for the singular perturbations, it is impossible to operate directly with the difference of operators $A_{1}$ and $A$. 

At the same time, the known properties of the primal operator $A$ and Krein's resolvent formula \cite{Kr} for $R_{1}(z)$ contain exhaustive information on the spectrum and scattering associated with the singular perturbation $A_{1}$. We would like to emphasize that this information can be extracted directly using the following representation of Krein's resolvent formula for singular perturbations \cite{Ad} without investigating the complexities of extension theory beforehand.
  
\begin{thm}\label{kreinnext1}
	Let $\mathcal{H}$ and $\mathcal{K}$ be Hilbert spaces, $A$ be an unbounded selfadjoint operator in $\mathcal{H}$ and  $R(z),\; \mathrm{Im}z\neq 0,$ is the resolvent of $A$. Let $G(z)$  be a bounded holomorphic in the open upper and lower half-planes operator function from $\mathcal{K}$ to  $\mathcal{H}$ satisfying the conditions 
	\begin{itemize}
		\item{for any non-real $z,z_{0}$
			\begin{equation}\label{res3}
			G(z)=G(z_{0})+(z-z_{0})R(z)G(z_{0}); 
			\end{equation}
		}
		\item{zero is not an eigenvalue of the operator  $G(z)^{*}G(z)$ at least for one and hence for all non-real $z$;}  \item {the intersection of the domain $\mathcal{D}(A)$ of $A$ and the subspace \newline $\mathcal{N}=\overline{G(z_{0})\mathcal{K}}\subset\mathcal{H}$ consists only of the zero-vector,}
	\end{itemize}
	and let $Q(z)$ be holomorphic in the open upper and lower half-planes operator function in $\mathcal{K}$ such that 
	\begin{itemize}
		\item{$Q(z)^{*}=Q(\bar{z}), \; \mathrm{z}\neq 0$;}
		\item{for any non-real $z,z_{0}$
			\begin{equation}
			Q(z)-Q(z_{0})=(z-z_{0})G(\bar{z_{0}})^{*}G(z). 
			\end{equation}
		}
	\end{itemize} 
	Then for any selfadjoint operator $L$ in $\mathcal{K}$ such that the operator $Q(z)+L, \:\mathrm{Im}z\neq 0, $ is boundedly invertible the operator function
	\begin{equation}\label{krein5}
	R_{L}(z)=R(z)-G(z)\left[Q(z)+L \right]^{-1} G(\bar{z})^{*}
	\end{equation}
	is the resolvent of some singular selfadjoint perturbation  $A_{L}$ of $A$. 
\end{thm}

A singular perturbation $A_{1}$ of a given selfadjoint operator $A$ is  \newline from now on referred to as \textit{close} (to {A}) if the difference of the resolvents $R_{1}(z)-R(z), \, \mathrm{Im}z\neq 0,$ is a trace class operator. Per this definition, the absolutely continuous components of the close perturbation $A_{1}$ and the operator $A$ are similar \cite{Kat}. Moreover, absolutely continuous components of close singular perturbations can be described using the constructions of scattering theory. 

The aim of this study is \begin{itemize}
	\item {to derive  an explicit expression relating the scattering matrix for an unbounded selfadjoint operator $A$ with absolutely continuous spectrum and its close singular perturbation $A_{1}$ with operator $L$ and the function $Q(z)$ in the Krein's formula for the same pair $A,A_{1}$, slightly modifying the approach developed in \cite{AdP};} 
		\item {illustrate the corresponding results by the example of singular selfadjoint perturbations of the Laplace operator $$A=-\Delta=-\frac{\partial^{2}}{\partial x^{2}_{1}}-\frac{\partial^{2}}{\partial x^{2}_{2}}-\frac{\partial^{2}}{\partial x^{2}_{3}} $$  in $\mathbf{L}_{2}(\mathbf{R}_{3})$  defined on the Sobolev subspaces ${H}_{2}^{2}\left( \mathbf{R}_{3}\right)$.}
\end{itemize}

In Section 2, for completeness, we give a sketch of the proof of Theorem \ref{kreinnext1} and indicate additional conditions on the selfadjoint $A$ and its singular perturbation $A_{1}$  under which the difference of resolvents $R_{1}(z)-R(z)$  is a trace class operator, that is the conditions under which the singular perturbation  $A_{1}$ and the unperturbed operator $A$ are close. These conditions are crucial as they determine the applicability of the scattering theory to the given perturbation \cite{Ya}. Further, these conditions are refined for singular perturbations of the Laplace operator in $\mathbf{L}_{2}(\mathbf{R}_{3})$.

In Section 3, which is some modified version of the paper \cite{AdP}, we calculate the scattering matrix for a self-adjoint operator $A$ with an absolutely continuous spectrum and its close singular perturbation $A_{1}$ and find an explicit expression of the scattering matrix for the pair $A_{1},A$ through the entries in Krein's formula (\ref{krein5}). It should be noted that a similar expression for the scattering matrix was derived in terms of the abstract Titchmarsh-Weyl operator function in \cite{BMN},\cite{BMN1}. 

Section 4 explores scattering theory for self-adjoint singular perturbations of the standard Laplace operator $A$  in $L^2(\mathbb{R}^3)$, modeled as an infinite series of zero-range potentials. By imposing specific conditions on the spatial arrangement and "heights" or "depths" of these potentials, which provide nuclear proximity of the resolvents of $A_1$ and $ A $, and using results from the previous section, we derive an explicit expression for the scattering matrix for the pair $A_1, A $. In line with the fact that the absolutely continuous components of both operators $A_1$ and $A$ have a uniform infinitely multiple Lebesgue spectrum covering the half-axis $[0,\infty)$, the values of the obtained scattering matrix represent an operator function of the spectral parameter on $[0,\infty)$ the values of which are bounded integral operators in space $L^2(\mathbf{S}_2)$, $\mathbf{S}_{2}$ be the unit sphere in $\mathbf{R}_{3}$. Under the specified conditions, this operator function appears to be continuous on $(0,\infty)$ with respect to the operator norm.  

\section{Resolvents of singular selfadjoint perturbations of Laplace operator}
To make the paper self-contained, we outline the proof of Theorem \ref{kreinnext1}. 

Recall the following well-known criterion (see, for example \cite{Ad}):\textit{ a holomorphic function $R(z)$ on the open upper and lower half-plains of the complex plane whose values are bounded linear operators in Hilbert space $\mathcal{H}$ is the resolvent of a selfadjoint operator $A$ in $\mathcal{H}$ if and only if 
 \begin{itemize}
	\item{\begin{equation}\label{first1}
		\ker R(z)=\{0\};
		\end{equation}}
	\item{\begin{equation}\label{first}
		R(z)^{*}=R(\overline{z});
		\end{equation}}
	\item{for any non-real $z_{1},z_{2}$ the Hilbert identity 
		\begin{equation}\label{hilb1}
		R(z_{1})-R(z_{2})=(z_{1}-z_{2})R(z_{1})R(z_{2})
		\end{equation} }
	 \end{itemize}
holds.}  

So to verify the validity of Theorem 1.1, it is only necessary to check that its assumptions about $R_{L}(z)$ guarantee the fulfillment of the conditions (\ref{first1}) - (\ref{hilb1}). 

But relation (\ref{first}) and identity (\ref{hilb1}) for $R_{L}(z)$ directly follow from the fact that $R(z)$ is the resolvent of the self-adjoint operator and the properties, which $G(z)$ and $Q(z)$ possess according to the assumptions of Theorem \ref{kreinnext1}. 

To see that (\ref{first1}) is true for $R_{L}(z)$ suppose that there is a vector $h\in \mathcal{H}$  such that $R_{L}(z)h=0$ for some non-real $z$. By (\ref{krein5}) this means that
\begin{equation}\label{krein5a}
R(z)h=G(z)\left[Q(z)+L \right]^{-1} G(\bar{z})^{*}h.
\end{equation}
But by assumption, $R(z)\mathcal{H}\cap G(z)\mathcal{K}=\{0\}$ for non-real $z$. Hence $R(z)h=0$. Given that $R(z)$ is the resolvent of a self-adjoint operator, we conclude from this that $h=0$. Thus, the statement of Theorem \ref{kreinnext1} is true. 

Turning to the description of singular perturbations of the Laplace operator $A=-\Delta$ in $\mathbf{L}_{2}(\mathbf{R}_{3})$ remind that its resolvent $R(z)$ is the integral operator
\begin{equation}\label{lapl} \begin{array}{c}
\left(R(z)f\right)(\mathbf{x})=\frac{1}{4\pi}
\int_{\mathbf{R}_{3}}\frac{e^{i\sqrt{z}|\mathbf{x}-\mathbf{x}^\prime|}}{|\mathbf{x}-\mathbf{x}^\prime|}f\left(\mathbf{x}^\prime \right)d\mathbf{x}^\prime , \quad \mathrm{Im}{\sqrt{z}}>0, \\ \mathbf{x}=\left(x_{1},x_{2},x_{3} \right) , 
f(\cdot)\in {\mathbf{L}_{2}(\mathbf{R}_{3})}.
\end{array}
\end{equation} 
In the simplest case, when the support of the singular perturbation of the Laplace operator is a finite set of points $\mathbf{x}_{1},...,\mathbf{x}_{N}$ of $\mathbf{R}_{3}$, , the role of the auxiliary space $\mathcal{K}$ in Krein's formula (\ref{krein5}) for the perturbed resolvent $R_{L}(z)$ is natural to provide the linear space $\mathbf{C}_{N}$ and as a corresponding linear mapping $G(z)$ of $\mathbf{C}_{N}$ into $\mathbf{L}_{2}(\mathbf{R}_{3})$ may be taken the operator function, which transforms vectors of the canonical orthonormal basis in $\mathbf{C}_{N}$: $$
\mathbf{e}_{1}=\left( \begin{array}{c} 1 \\ 0 \\ ...\\ 0 \end{array}\right), \dots ,\mathbf{e}_{N}=\left( \begin{array}{c} 0 \\  ... \\0 \\1 \end{array}\right),
$$ respectively, into the functions
\begin{equation}\label{simp1}
\left( G(z)\mathbf{e}_{n}\right)(\mathbf{x})=g_{n}(z;\mathbf{x})= R(z)\delta(\cdot -\mathbf{x}_{n})(\mathbf{x})=\frac{1}{4\pi}\frac{e^{i\sqrt{z}|\mathbf{x}-\mathbf{x}_{n}|}}{|\mathbf{x}-\mathbf{x}_{n}|}, \quad 1\leq n\leq N.
\end{equation}
 If, in this case, the matrix function 
\begin{equation}\label{simp2}
Q(z)=\left(q_{mn}(z) \right)_{m,n=1}^{N}=\left\lbrace \begin{array}{c} 
q_{mn}(z)=g_{n}(z;\mathbf{x}_{m}-\mathbf{x}_{n}), \quad m\neq n,\\
q_{mm}(z)=\frac{i\sqrt{z}}{4\pi}
\end{array}\right. . 
\end{equation}
is substituted into (\ref{kreinnext1}) as the $Q$-function, then for any invertible Hermitian matrix $L=\left( w_{mn} \right)_{m,n=1}^{N} $ the resulting operator function $R_{L}(z)$ satisfies all the conditions of Theorem \ref{kreinnext1} and hence appears to be the resolvent of a singular seladjoint perturbation of the Laplace operator $A_{L}$ . Setting  $\mathcal{N}=G(z_{0})\mathbf{C}_{N},\; \mathrm{Im}z_{0}\neq 0,$ one can easily deduce from the Krein formula for $R_{L}(z)$ that in this case $A_{L}$ is nothing else than the Laplace differential operator $-\Delta$ with the domain \cite{Ad} 
\begin{equation}\label{simp3a}
\begin{array}{c}      
\mathcal{D}_{L}:=  \left\lbrace f: \, f=f_{0}+g, \, f_{0}\in {H}_{2}^{2}\left(\mathbf{R}_{3}\right), g\in\mathcal{N}, \right. \\  
\underset{\rho_{m}\rightarrow 0}{\lim}\left[\frac {\partial} {\partial\rho_{m}}\left(\rho_{m}f(\mathbf{x})\right)\right]+\sum\limits_{n=1}^{N}w_{mn} 
\underset {\rho_{n}\rightarrow 0}{\lim}\,[\rho_{n}\,f(\mathbf{x})] =0, \\ \left. \rho_{n}=|\mathbf{x}-\mathbf{x}_{n}|, \quad 1\leq n\leq N  \right\rbrace .  \end{array} 
\end{equation}
 In other words, here we are dealing with a perturbation in the form of a finite sum of zero-range potentials \cite{BF, DeOst}. 
 
 For perturbations of the Laplace operator A in the form of infinite sum of zero-range potentials located at points $\{\mathbf{x}_{n}\}_{m,n\in \mathbb{Z}}$ under an additional condition 
 \begin{equation}\label{dist}
 	\underset{_{m,n\in \mathbb{Z}}}{\inf}\left|\mathbf{x}_{m}-\mathbf{x}_{n} \right|=d>0.   
 \end{equation}
 the set of boundary conditions of the form (\ref{simp3a})  also generates a singular self-adjoint perturbation $A_{L}$ of $A$. 
 
 In what follows, it is assumed that the elements of canonical basis $\left\lbrace   \mathbf{e}_{n}=\left(\delta_{nm} \right)_{m\in \mathbb{Z}}\right\rbrace $ of the Hilbert space $\mathbf{l}_{2}$ belong to the domains of all mentioned operators in this space, and no distinction is made between those operators and the matrices that represent them in the canonical basis of $\mathbf{l}_{2}$. 
 
  \begin{thm} [\cite{GHM}]\label{kreinmain1+}
  Let $Q(z)$ be the operator function in $\mathbf{l}_{2}$ generated by the infinite matrix 
 \begin{equation}\label{simp3}
 Q(z)=\left(q_{mn}(z) \right)_{m,n\in \mathbb{Z}}, \, q_{mn}(z)=\left\lbrace \begin{array}{c} 
 q_{mn}(z)=g(z;\mathbf{x}_{m}-\mathbf{x}_{n}), \quad m\neq n,\\
 q_{mm}(z)=\frac{i\sqrt{z}}{4\pi}
 \end{array}\right. . 
 \end{equation}
 and $L$ be a selfadjoint operator   defined by the infinite matrix $\left( w_{mn}\right)_{m,n\in \mathbb{Z}}$ in the space of bilateral sequences $\mathbf{l}_{2}$. If the condition (\ref{dist}) holds, then the operator function
 \begin{equation}\label{krein6a} \begin{array}{c}
  R_{L}(z)=R(z)-\sum\limits_{m,n\in \mathbb{Z}}^{}\left([ Q(z)+4\pi L]^{-1}\right)_{mn} \left(\cdot,g_{n}(\bar{z};\cdot) \right)g_{m}(z;\cdot), \\ g_{n}(z;\mathbf{x})= R(z)\delta(\cdot -\mathbf{x}_{n})(\mathbf{x}), \end{array}
 \end{equation} 
 is the resolvent of the selfadjoint operator $A_{L}$ in $\mathbf{L}_{2}(\mathbf{R}_{3})$, which is the Laplace operator with the domain
 \[
 \mathcal{D}_{L}:=  \left\lbrace f: \, f=f_{0}+g, \, f_{0}\in {H}_{2}^{2}\left(\mathbf{R}_{3}\right), \, g\in \mathcal{N}, \right.  \] 
 
 \[ \left. \begin{array}{c} \underset{\rho_{m}\rightarrow 0}{\lim}\left[\frac {\partial} {\partial\rho_{m}}\left(\rho_{m}f(\mathbf{x})\right)\right]+\sum\limits_{n\in \mathbb{Z}}^{}w_{mn} \underset
 {\rho_{n}\rightarrow 0}{\lim}\,[\rho_{n}\,f(\mathbf{x})] =0, \\ \rho_{n}=|\mathbf{x}-\mathbf{x}_{n}|, \quad n\in \mathbb{Z} .  \end{array}\right.  \]
 \end{thm}
 The above results on perturbations of the Laplace operator, which are reducible to a finite or infinite sum of zero-range potentials, contain nothing new compared with those collected quite long ago in books\cite {AGHH, AK}. However, they follow directly from Theorem 1 without recourse to the theory of self-adjoint extensions of symmetric operators (see\cite{Ad}).

Of course, the sparseness (\ref{dist}) of the set $\{\mathbf{x}_{n}\}_{n\in \mathbb{Z}} $ is not necessary for an expression of the form (\ref{simp3}) to represent the resolvent of a self-adjoint perturbation of the Laplace operator by an infinite sum of zero-range potentials. If (\ref{dist}) is not satisfied, then the operator functions $R_{L} (z)$, defined as in Theorem \ref{kreinmain1+}, may still be resolvents of close singular selfadjoint perturbations $A_{L}$ of $A$ for an account of the particular choice of parameters $L$. 
\begin{thm}\label{help1}
Let $\{\mathbf{x}_{n}\}_{n\in \mathbb{Z}}$ be a sequence of different points of $\mathbf{R}_{3}$ such that each compact domain of $\mathbf{R}_{3}$ contains a finite number of accumulation points of this set. Put
\begin{equation}\label{help2a}
\delta_{0}=|\mathbf{x}_{0}|; \; \delta_{n}=\underset{-n\leq j\neq k \leq n}{\min}|\mathbf{x}_{j}-\mathbf{x}_{k}|,  \; n=1,2,... 
\end{equation} 
Let $L$ be a selfadjoint operator in $\mathbf{l}_{2}$ such that the point $z=0$ belongs to the resolvent set of $L$ and and the Hermitian matrix $\left( b_{mn}\right)_{m,n\in \mathbb{Z}}$ of non-negative operator $|L|^{-\frac{1}{2}}$ in the canonical basis $\left\lbrace \mathbf{e}_{n}=\left(\delta_{nm} \right)_{m\in \mathbb{Z}} \right\rbrace_{n\in \mathbb{Z}} $ of $\mathbf{l}_{2}$ satisfies the conditions:
\begin{equation}\label{cond1}
\sum_{n\in \mathbb{Z}}\left(\sum_{m\in \mathbb{Z}}\left|b_{nm}\right|\right)^{2}<\infty, \quad  \sum_{n\in \mathbb{Z}}\left(\sum_{m\in \mathbb{Z}}\left|b_{nm}\right|\frac{1}{\delta_{m}}\right)^{2}<\infty.
\end{equation}
Then the operator function
\begin{equation}\label{krein7} \begin{array}{c}
		R_{L}(z)=R(z)-\sum_{m,n\in \mathbb{Z}}\left([ Q(z)+4\pi L]^{-1}\right)_{mn} \left(\cdot,g_{n}(\bar{z};\cdot) \right)g_{m}(z;\cdot) \\ =
	R(z)-\tilde{G}(z)\left([ \tilde{Q}(z)+4\pi J_{L}]^{-1}\right)\tilde{G}(\bar{z})^{*},\\ 	
		g_{n}(z;\mathbf{x})= R(z)\delta(\cdot -\mathbf{x}_{n})(\mathbf{x}), \\ \tilde{G}(z)\mathbf{h}=\sum_{n}\left(|L|^{-\frac{1}{2}}\mathbf{h},\mathbf{e}_{n}\right)g_{n}(z;\cdot), \, \mathbf{h}\in \mathbf{l}_{2},\\
		\tilde{Q}(z)=|L|^{-\frac{1}{2}}Q(z)|L|^{-\frac{1}{2}}, \quad J_{L}=L\cdot |L|^{-1}, \quad \mathrm{Im}z\neq 0,
	\end{array}
\end{equation}
is the resolvent of a close singular perturbation $A_{L}$ of the Laplace operator $A$ 
with the domain
\begin{equation}\begin{array}{c}
\mathcal{D}_{L}:=  \left\lbrace f: \, f=f_{0}+g, \, f_{0}\in {H}_{2}^{2}\left(\mathbf{R}_{3}\right),\right\rbrace \, g\in \mathcal{N},  \\

  \underset{\rho_{m}\rightarrow 0}{\lim}\left[\frac {\partial} {\partial\rho_{m}}\left(\rho_{m}f(\mathbf{x})\right)\right]+4\pi\sum\limits_{n\in \mathbb{Z}}^{}\left( L\mathbf{e}_{n},\mathbf{e}_{m}\right)  \underset
	{\rho_{n}\rightarrow 0}{\lim}\,[\rho_{n}\,f(\mathbf{x})] =0, \\ \rho_{n}=|\mathbf{x}-\mathbf{x}_{n}|, \quad n\in \mathbb{Z} .  \end{array} \end{equation}
\end{thm}
\begin{proof}

First of all, note that the mapping $\tilde{G}(z)$ from $\mathbf{l}_{2}$  to ${L}_{2}\left(\mathbf{R}_{3}\right)$  is a Hilbert-Schmidt operator. Indeed, for the canonical basis $\left\lbrace \mathbf{e}_{n} \right\rbrace_{n\in \mathbb{Z}} $ in $\mathbf{l}_{2}$  we have
\begin{equation*}\begin{array}{c}
		\left\|\tilde{G}(z)\mathbf{e}_{n}\right\|^{2}=\left(|L|^{-\frac{1}{2}}G(z)^{*}G(z) |L|^{-\frac{1}{2}}\mathbf{e}_{n},\mathbf{e}_{n}\right) \\=\sum_{m,m^\prime\in \mathbb{Z}}b_{nm^\prime}\cdot\frac{1}{8\pi}e^{-\mathrm{Im}\sqrt{z}|\mathbf{x}_{m}-\mathbf{x}_{m^\prime}|}\frac{\sin{\mathrm{Re}\sqrt{z}|\mathbf{x}_{m}-\mathbf{x}_{m^\prime}|}}{\mathrm{Re}\sqrt{z}|\mathbf{x}_{m}-\mathbf{x}_{m^\prime}|}\cdot b_{mn} \\
		\leq \frac{1}{8\pi}\sum\limits_{m,m^\prime\in \mathbb{Z}}|b_{nm}|\cdot |b_{m^\prime n}| \leq \frac{1}{8\pi}\left( \sum\limits_{m\in \mathbb{Z}}|b_{mn}|\right)^{2}<\infty. 
	\end{array}
\end{equation*}
and by virtue of (\ref{cond1}) we have
\begin{equation}
		\sum_{n\in \mathbb{Z}}\left\|\tilde{G}(z)\mathbf{e}_{n}\right\|^{2}
		\leq \frac{1}{8\pi}\sum_{n\in \mathbb{Z}}\left(\sum_{m\in \mathbb{Z}}|b_{mn}|\right)^{2}<\infty. 
	 \label{nucl}
\end{equation}
 The inequality (\ref{nucl}), however, is a criterion for an operator in $\mathbf{l}_{2}$ to belong to the Hilbert-Schmidt class. 

Turning to the operator function $\tilde{Q}(z)$, we represent the corresponding matrix function as the sum $D(z)+M(z)$ of matrices
\begin{equation}\label{diag1}
	D(z)=\frac{i\sqrt{z}}{4\pi}\left( \sum_{j\in \mathbb{Z}}b_{mj}b_{jn}\right)_{m,n\in\mathbb{Z}}   
\end{equation} 
and 
\begin{equation}\label{rest1}
		M(z)=\left(u_{mn}(z)=\sum_{n^\prime,m^\prime \in \mathbb{Z}}b_{n n^\prime}g(z;\mathbf{x}_{n^\prime}-\mathbf{x}_{m^\prime})b_{m^\prime m}, \, m^\prime\neq n^\prime \right)_{m,n\in \mathbb{Z}} .
\end{equation}
  	
  	Note that, according to the first assumption in (\ref{cond1}), $|L|^{-\frac{1}{2}} $ is a Hilbert-Schmidt operator, since $$\sum_{n,j\in \mathbb{Z}}b_{nj}b_{jn}=\sum_{n,j\in \mathbb{Z}}|b_{nj}|^{2} \leq \sum_{n\in \mathbb{Z}}\left( \sum_{j\in \mathbb{Z}}|b_{nj}|\right) ^{2}<\infty.$$
As follows, $D(z)=\frac{i\sqrt{z}}{4\pi}|L|^{-1}$ is a trace class operator.

Taking into account, further, that
\begin{equation}\label{imp1}
\begin{array}{c}
\left|\sum\limits_{m^\prime\neq n^\prime }g(z;\mathbf{x}_{n^\prime}-\mathbf{x}_{m^\prime})b_{m^\prime m}, \,   \right|\leq\sum\limits_{m^\prime\neq n^\prime }\frac{1}{4\pi\left| \mathbf{x}_{n^\prime}-\mathbf{x}_{m^\prime}\right|}|b_{m^\prime m}| \\ 
\leq \frac{1}{4\pi \delta_{n^\prime}}\sum\limits_{|m^{\prime}|\leq|n^{\prime}|}|b_{m^\prime m}|+\sum\limits_{|m^{\prime}|>|n^{\prime}|}\frac{1}{4\pi \delta_{m^\prime}}|b_{m^\prime m}|\leq \frac{1}{\delta_{n^\prime}}C^{(0)}(m)+C^{(1)}(m)_, \\
C^{(0)}(m)=\sum\limits_{j\in \mathbb{Z}}\left| b_{mj}\right|, \quad C^{(1)}(m)=\frac{1}{4\pi}\sum\limits_{j\in \mathbb{Z}}\left| b_{mj}\right|\frac{1}{\delta_j },
\end{array}
\end{equation} 	
and that by (\ref{cond1})
\begin{equation}\label{imp2}
	\sum_{m\in \mathbb{Z}}\left[C^{(0)}(m)\right]^{2}<\infty,    \quad \sum_{m\in \mathbb{Z}}\left[C^{(1)}(m)\right]^{2}<\infty, 
\end{equation}
we conclude that
\begin{equation}\label{imp3}
	\left|u_{nm} (z)\right|\leq \left[C^{(1)}(n)C^{(0)}(m)+C^{(1)}(m)C^{(0)}(n) \right] ,  
\end{equation}
and, therefore
\begin{equation}
\sum_{nm\in \mathbb{Z}}	\left|u_{nm} (z)\right|^{2}<\infty.
\end{equation}
We see that $M(z)$ is a Hilbert-Schmidt operator, and hence the sum $\tilde{Q}(z)=D(z)+M(z)$ is at least a Hilbert-Schmidt operator.

Remember now that $\tilde{Q}(z)+J_{L} $ is an operator function of the Nevanlinna class and that for any $\mathbf{h}\in \mathbf{l}_{2}$ and any non-real $z$ we have
\begin{equation}\label{imp4}\begin{array}{c}
\frac{1}{\mathrm{Im}z}\cdot \mathrm{Im}\left[\left(\tilde{Q}(z)\mathbf{h},\mathbf{h}  \right)+\left( J_{L}\mathbf{h},\mathbf{h} \right) \right] \\ =\frac{1}{\left( 2\pi\right)^{3} }\int_{\mathbf{R}_{3}} \frac{1}{\left(|\mathbf{k}|^{2}-\mathrm{Re}z \right)^{2} +\left( \mathrm{Im}z\right)^{2} }\left|\sum_{n\in \mathbb{Z}}\left(|L|^{-\frac{1}{2}}\mathbf{h},\mathbf{e}_{n} \right)e^{-\mathbf{k}\cdot\mathbf{x}_{n}}\right|^{2} d\mathbf{x}\geq 0
\end{array}		
\end{equation}	

The function 
\begin{equation}\label{imp5}
	\hat{\mathbf{h}}(\mathbf{k})=\sum_{n\in \mathbb{Z}}\left(|L|^{-\frac{1}{2}}\mathbf{h},\mathbf{e}_{n} \right)e^{-\mathbf{k}\cdot\mathbf{x}_{n}}, \; \mathbf{h}\in \mathbf{l}_{2},
\end{equation}
that emerges in (\ref{imp4}) is bounded and continuous. Indeed, setting
\begin{equation*}
    \mathbf{h}=\sum\limits_{m\in \mathbb{Z}}a_{m}\mathbf{e}_{m}
\end{equation*}
and, taking into account the equality $b_{nm}=\overline{b_{mn}}$ we see that
\begin{equation}\begin{array}{c}\label{bohr}
\sum\limits_{n\in \mathbb{Z}}\left| \left(|L|^{-\frac{1}{2}}\mathbf{h},\mathbf{e}_{n} \right)\right|= \sum\limits_{n\in \mathbb{Z}}\left|\sum\limits_{m\in \mathbb{Z}}b_{nm}a_{m}\right| \\ \leq\sum\limits_{m\in \mathbb{Z}} \left( \sum\limits_{n\in \mathbb{Z}}\left|b_{mn} \right| \right)\left|a_{m}\right|  \leq \sqrt{\sum\limits_{m\in \mathbb{Z}} \left( \sum\limits_{n\in \mathbb{Z}}\left|b_{mn} \right| \right)^{2}}\cdot \left\| \mathbf{h}\right\| <\infty.	
\end{array}
\end{equation} 
Note for an expanding family of domains $$ \mathbf{\Omega}_{K}=\left\lbrace\mathbf{k}:-K<k_{1},k_{2},k_{3}<K \right\rbrace \subset \mathbf{R}_{3} $$ and any $n\in \mathbb{Z}$ by virtue of (\ref{imp5}) and (\ref{bohr}) we have
\begin{equation}\label{bohr1}
\underset{K\rightarrow \infty}{\lim}\frac{1}{8K^{3}}\int_{\mathbf{\Omega}_{K}}e^{i\mathbf{k}\cdot\mathbf{x}_{n}}\hat{\mathbf{h}}(\mathbf{k})d\mathbf{k}=\left(|L|^{-\frac{1}{2}}\mathbf{h},\mathbf{e}_{n} \right), \, n\in \mathbb{Z}.
\end{equation}

Now we can conclude that the equality sign in (\ref{imp4}) is reached if and only if $h(k) \equiv 0$, and in view of (\ref{bohr1}), if and only if $|L|^{-\frac{1}{2}}\mathbf{h}=0$.  By our assumptions, the latter means that $\mathbf{h} = 0$. Therefore, "0" cannot be an eigenvalue of the operator $\tilde{Q}(z)+J_{L}, \, \mathrm{Im}z \neq 0$, and the same is true for the operator $J_{L}\cdot\tilde{Q}(z)+I$. In other words, "-1" is not an eigenvalue of the operator $J_{L}\cdot\tilde{Q}(z)$. But $J_{L}\cdot\tilde{Q}(z)$ is a compact operator. As follows, $J_{L}\cdot\tilde{Q}(z)+I$ boundedly invertible and so is the operator $\tilde{Q}(z)+J_{L}$, $$\left[\tilde{Q}(z)+J_{L} \right]^{-1}=J_{L}\cdot \left[J_{L}\cdot\tilde{Q}(z)+I \right]^{-1} .$$
We proved that for non-real $z$ the values of operator function $\left[\tilde{Q}(z)+J_{L} \right]^{-1}$ are bounded operators.   

Suppose further that for some nonreal $z$ there is a vector $f\in \mathbf{L}_{2}(\mathbf{R}_{3})$ from the linear set $\tilde{G}(z)\mathbf{l}_{2}$ that belongs to the domain $\mathfrak{D}(A)$ of the Laplace operator $A$. This vector as any vector from $\mathfrak{D}(A)$ can be represented in the form $f=R(z)w$ with some $w\in\mathbf{L}_{2}(\mathbf{R}_{3})$ while by our assumption there is a vectors $\mathbf{h}\in \mathbf{l}_{2}$ such that
\begin{equation}\label{lim-a}
	R(z)w-\tilde{G}(z)\mathbf{h}=0.
\end{equation}   

Now recall that for each $w\in\mathbf{L}_{2}(\mathbf{R}_{3})$ and any infinitesimal $\varepsilon>0$, it is possible to find an infinitely smooth compact function $\phi(\mathbf{x})$ which is also equal to zero at some $\eta$-neighborhood of all the accumulation points and isolated points of the set $\left\lbrace \mathbf{x}_{n} \right\rbrace$ on the support of $\phi(\mathbf{x})$to satisfy the condition   
\begin{equation}\label{lim-b}
	\left| \left(w,\phi \right)_{\mathbf{L}_{2}(\mathbf{R}_{3})} \right|\geq\left(1-\varepsilon \right)\left\| w \right\|^{2}_{\mathbf{L}_{2}(\mathbf{R}_{3}}.
\end{equation}
Taking into account further that for $\phi(\mathbf{r})$, as well as for any smooth compact function,
\begin{equation}\label{laplres}
	\phi(\mathbf{x})= \frac{1}{4\pi}
	\int_{\mathbf{R}_{3}}\frac{e^{i\sqrt{z}|\mathbf{x}-\mathbf{x}^\prime|}}{|\mathbf{x}-\mathbf{x}^\prime|}\left[ -\Delta\phi\left(\mathbf{x}^\prime \right)-z\phi\left(\mathbf{x}^\prime \right)\right]  d\mathbf{x}^\prime, 
\end{equation}
we notice that 
\begin{equation*} \begin{array}{c}
		\left(G(z)w,\left[-\Delta\phi -\bar{z}\phi \right]\right)_{\mathbf{L}_{2}(\mathbf{R}_{3})}=\left(w,\phi\right)_{\mathbf{L}_{2}(\mathbf{R}_{3})}, \\  \left(\tilde{G}(z)\mathbf{h},\left[-\Delta\phi -\bar{z}\phi \right]\right)_{\mathbf{L}_{2}(\mathbf{R}_{3})} =0, \quad \mathbf{h}\in \mathbf{l}_{2}.		
	\end{array}
\end{equation*}
Hence by virtue of (\ref{lim-b}) we conclude that 
\begin{equation}\label{lim-c}
	\begin{array}{c}
		\|R(z)w-\tilde{G}(z)\mathbf{h}\|_{\mathbf{L}_{2}(\mathbf{R}_{3})}\cdot
		\left\| -\Delta\phi-\bar{z}\phi  \right\|_{\mathbf{L}_{2}(\mathbf{R}_{3})} \\ \geq 
		\left|\left( \left[R(z)w-\tilde{G}(z)\mathbf{h}\right],\left[-\Delta\phi-\bar{z}\phi \right] 
		\right)_{\mathbf{L}_{2}(\mathbf{R}_{3})}\right|  =\left|\left(w,\phi\right)_{\mathbf{L}_{2}(\mathbf{R}_{3})}\right| \\ \geq\left(1-\varepsilon\right)
		\left\| w\right\|^{2}_{\mathbf{L}_{2}(\mathbf{R}_{3})}.
	\end{array}
\end{equation}
But given (\ref{lim-a})  the last inequality in (\ref{lim-c}) must necessarily be violated unless  $ w=0 $. Therefore  
\begin{equation}\label{imp6}
\tilde{G}(z)\mathbf{l}_{2} \cap \mathfrak{D}(A)=\{0\}.	
\end{equation}

From (\ref{imp6}) and the established properties of the operator function $R_{L}(z)$ it follows, in particular, that $\ker{R_L(z)}=\{0\}, \, \mathrm{Im}z\neq 0 $. Indeed by virtue of (\ref{imp6}) those properties the equality $$R(z)w-\tilde{G}(z)\left([ \tilde{Q}(z)+4\pi J_{L}]^{-1}\right)\tilde{G}(\bar{z})^{*}w=0, \; \mathrm{Im}z\neq 0, $$ for some $w\in \mathbf{L}_{2}(\mathbf{R}_{3}$ implies $w=0$ .

Verification of the remaining conditions of Theorem \ref{kreinnext1} for $R_{L}(z)$ is not difficult and is left to the reader.
\end{proof}

\section{Scatttering matrices}

From this point onward, we will assume that the spectrum of the selfadjoint operator $A$ is absolutely continuous. Unless stated otherwise, we will consider that the operator functions $G(z)$ and $Q(z)$ and the operator $L$ in the expression (\ref{krein5}) are such that all the conditions of Theorem \ref{kreinnext1} hold. Additionally, we will assume that the difference of the resolvent of selfadjoint perturbation $A_{L}$ of $A$ and the resolvent of $A$ is a trace-class operator. By virtue of these assumptions (see, for example,\cite{Ya}) the wave operators $A_{\pm}\left(A_{L},A \right)$ defined as strong limits
\begin{equation}\label{wave1}
\mathfrak{W}_{\pm}\left(A_{L},A \right) =s-\underset{t\rightarrow\pm\infty}{\lim}e^{iA_{L}t}e^{-iAt} 
\end{equation}
 exist and are isometric mappings of $\mathcal{H}$ onto the absolutely continuous subspace of $A_{L}$. Recall that in this case the wave operators $\mathfrak{W}_{\pm}\left(A_{L},A \right)$ are intertwining for the spectral functions $E_{\lambda}, E^{(L)}_{\lambda}, \, -\infty<\lambda<\infty , $  of the operators $A,A_{L}$ in the sense that $$\mathfrak{W}_{\pm}\left(A_{L},A \right)E_{\lambda}=E^{(L)}_{\lambda}\mathfrak{W}_{\pm}\left(A_{L},A \right).$$
 
   The scattering operator, which is defined as the  product of  wave operators   
 \begin{equation}\label{scat1}
  \mathfrak{S}(A_{L},A)=\mathfrak{W}_{+}\left(A_{L},A \right)^{*}\mathfrak{W}_{-}\left(A_{L},A \right)
   \end{equation} 
  is an isometric operator in $\mathcal{H}$ and
\begin{equation}\label{commut}
	 E_{\lambda}\mathfrak{S}(A_{L},A)=\mathfrak{S}(A_{L},A)E_{\lambda}, \; -\infty<\lambda<\infty .
\end{equation}
 
 Therefore for the representation of $A$ in 
$\mathcal{H}$ as the multiplication operator by $\lambda$ in the direct integral of Hilbert spaces 
$\mathfrak{h}(\lambda)$, $$\mathcal{H}\Rightarrow\int \limits_{-\infty}^{\infty}\oplus\mathfrak{h}(\lambda)d\lambda,$$
the scattering operator $\mathfrak{S}(A_{L},A)$ acts as the multiplication operator by a contractive operator  
function $S(A_{L},A)(\lambda)$, which will be below referred to as the scattering matrix.

These rather general assertions are specified below for close singular perturbations.
 
 Let $\sigma (A)$ denote the spectrum of $A$ and  $\Delta$ is some interval that $\subseteq\sigma(A)$.  We will assume that the part of $A$ on $E(\Delta)\mathcal{H}$ 
 has the Lebesgue spectrum of multiplicity $\mathfrak{n}\leq\infty$.  
 Then there exists an isometric operator $\mathfrak{F}$, which maps $E(\Delta)\mathcal{H}$ onto the space 
 $\mathbb{L}^{2}(\Delta;\mathcal{N})$ of vector function on $\Delta$ with values in the auxiliary Hilbert space $\mathcal{N}, \, \dim\mathcal{N}=\mathfrak{n},$  and such that 
 $\mathfrak{F}A|_{E(\Delta)\mathcal{H}}\mathfrak{F}^{-1}$ is the multiplication operator by independent variable in 
 $\mathbb{L}^{2}(\Delta;\mathcal{N})$. Using the notation 
 $$\mathfrak{F}(E(\Delta)f)(\lambda)=\mathbf{f}(\lambda), \quad f\in \mathcal{H}, \; \mathbf{f}(\cdot)\in  
 \mathbb{L}^{2}(\Delta;\mathcal{N}), $$ for any $f$ from the domain of $A$ we can write $$\mathfrak{F}\left(E(\Delta)Af\right)(\lambda)=\lambda\cdot\mathbf{f}(\lambda) $$
 and by (\ref{commut})  for any $f,g\in E(\Delta)\mathcal{H}$ we have
 \begin{equation*} \begin{array}{c}
 	(f,g)_{\mathcal{H}}=\int_{\Delta}\left(\mathbf{f}(\lambda),\mathbf{g}(\lambda) \right) _{\mathcal{N}}d\lambda, \\
 (\mathfrak{S}(A_{L},A)f,g)_{\mathcal{H}}=\int_{\Delta}\left(S(A_{L},A)(\lambda)\mathbf{f}(\lambda),\mathbf{g}(\lambda) \right) _{\mathcal{N}}d\lambda.	
\end{array} \end{equation*}
\begin{thm}\label{subopzer}
Let $A$ be a selfadjoint operator in Hilbert space $\mathcal{H}$ with an absolutely continuous spectrum, $A_{L}$ be its close selfadjoint perturbation, and the resolvent $R_{L}(z)$ of $A_{L}$ admits the representation (\ref{krein5}) with selfadjoint operator $L$ in Hilbert space $\mathcal{K}$,  functions $G(z)$ and $Q(z)$ with values being bounded operators from $\mathcal{K}$ to $\mathcal{H}$ and in $\mathcal{K}$, respectively satisfy all the conditions of Theorem \ref{kreinnext1}.  

Suppose additionally that \begin{itemize} 
\item {for fixed $\gamma>0$, some linearly independent system $\lbrace g_{k}\rbrace_{k\in \mathbb{Z}}\subset \mathcal{H}$ such that $\underset{k\in{\mathbb{Z}}}{\sup}\|g_{k}\|^{2}=M<\infty$, orthonormal system $\lbrace \mathbf{h}_{j}\rbrace_{j\in \mathbb{Z}}\subset \mathcal{K}$ and a set of numbers $\lbrace b_{jk}\rbrace_{jk\in \mathbb{Z}} $ such that $$\sum_{j\in\mathbb{Z}} \left(\sum_{k\in\mathbb{Z}}|b_{jk}|\right)^{2}<\infty $$
the operator $G(-i\gamma)$ admits the representation 
\begin{equation}\label{G1}
G(-i\gamma)=\sum_{j,k\in\mathbb{Z}}b_{jk}\left(\cdot\, ,\mathbf{h}_{j} \right)_{\mathcal{K}}g_{k} ;
\end{equation}}

\item{ on some interval $\Delta\subset\sigma(A)$ the operator $A$ has the Lebesgue spectrum of multiplicity $\mathfrak{n}\leq\infty$ and for the operator function $$\Gamma(z)=\left[|L|^{-\frac{1}{2}}Q(z)|L|^{-\frac{1}{2}}+J_{L} \right]^{-1}, \; J_{L}=L\cdot |L|^{-1}, \; \mathrm{Im}z\neq 0, $$
almost everywhere on $\Delta$ there is a weak limit
$$ \Gamma(\lambda +i0)=\underset{\varepsilon\downarrow 0}{\lim}\Gamma(\lambda +i\varepsilon), \; \lambda\in \Delta ,$$ such that
\begin{equation}\label{ess}
\underset{\lambda\in\Delta }{\mathrm{esssup}} \left\|\Gamma(\lambda +i0)\right\|<\infty.	
\end{equation}
 }
\end{itemize}
Then the image $\mathfrak{F}\mathfrak{S}(A_{L},A)\mathfrak{F}^{-1}$ of scattering operator acts in $\mathbb{L}^{2}(\Delta;\mathcal{N})$ as the multiplication by operator function
\begin{equation}\label{suboper1a}
\begin{array}{c}
	S(A_{L},A)(\lambda)=I-2i\sum\limits_{j,k\in\mathbb{Z}}\left( \Gamma(\lambda+i0)\mathbf{h}_{j}, \mathbf{h}_{k} \right)_{\mathcal{K}}\left(\cdot  ,	\hat{\mathbf{g}}_{j}(\lambda)\right)_{\mathcal{N}}\hat{\mathbf{g}}_{k}(\lambda),  \\ \hat{\mathbf{g}}_{j}(\lambda)=\sqrt{2\pi}(\lambda+i\gamma)\sum\limits_{l\in\mathbb{Z}}b_{jl}\left(\mathfrak{F}\mathbf{g}_{l}\right)(\lambda).
\end{array}
\end{equation} 				
\end{thm}
\begin{proof}

Notice that the existence of strong limits in (\ref{wave1}) ensures the validity 
of relations 
\begin{equation}\label{wave2}\begin{array}{c}
		\mathfrak{W}_{\pm}\left(A_{1},A \right) =s-\underset{\varepsilon\downarrow 
			0}{\lim}\:\varepsilon\int\limits_{0}^{\infty}e^{-\varepsilon t}e^{\pm iA_{L}t}e^{\mp iAt} \\ 
		=s-\underset{\varepsilon\downarrow 0}{\lim}\:\varepsilon\int\limits_{0}^{\infty}e^{-\varepsilon t}e^{\pm iA_{L}t}\int 
		\limits_{-\infty}^{\infty}e^{\mp i\lambda t}dE_{\lambda} =s-\underset{\varepsilon\downarrow 0}{\lim}\:\pm i\varepsilon 
		\int \limits_{-\infty}^{\infty}R_{L}(\lambda \pm i\varepsilon)dE_{\lambda}.
	\end{array}
\end{equation}
By (\ref{wave2}) the quadratic form of $\mathfrak{S}(A_{L},A)$ for any $f_{1},f_{2}\, \in \mathcal{H}$ can be written as follows
  \begin{equation}\label{scat2}
  \begin{array}{c}
  \left(\mathfrak{S}(A_{L},A)f_{1},f_{2} \right)=\left(\mathfrak{W}_{-}\left(A_{L},A \right)f_{1},\mathfrak{W}_{+}\left(A_{L},A 
  \right)f_{2} \right) \\ =\underset{\varepsilon ,\eta\downarrow 
  0}{\lim}\:-\eta\varepsilon\iint\limits_{-\infty}^{\infty} \left(R_{L}(\lambda + i\varepsilon)dE_{\lambda}f_{1}, R_{L}(\mu 
  - i\eta)dE_{\mu}f_{2}h\right)  \\ = \underset{\varepsilon ,\eta\downarrow 
  0}{\lim}\:-\eta\varepsilon\iint\limits_{-\infty}^{\infty} \frac{1}{\lambda-\mu+i(\varepsilon-\eta)}\left(\left[ 
  R_{L}(\lambda + i\varepsilon)-R_{L}(\mu + i\eta)\right]dE_{\lambda} f_{1}, dE_{\mu}f_{2}\right).
  \end{array}
  \end{equation}
 To proceed to the limits in (\ref{scat2}) first we note that 
\begin{equation}\begin{array}{c}
 \mp i\varepsilon\,R(\lambda \pm i\varepsilon)dE_{\lambda}f=f,\quad \varepsilon>0, \;f\in\mathcal{H},\\
 -\eta\varepsilon\iint\limits_{-\infty}^{\infty} \frac{1}{\lambda-\mu+i(\varepsilon-\eta)}dE_{\mu}\left[ R(\lambda + 
 i\varepsilon)-R(\mu + i\eta)\right]dE_{\lambda}f \\
=-\eta\varepsilon\iint\limits_{-\infty}^{\infty}dE_{\mu}R(\mu + i\eta)R(\lambda + 
i\varepsilon)dE_{\lambda}f=f.
\end{array}
\end{equation}
Then we use Krein's formula (\ref{krein5}) taking into account that by (\ref{res3}) for the fixed  $\gamma > 0$ the equality 
\begin{equation}\label{res3a}
G(z)=G(-i\gamma)+(z+i\gamma)R(z)G(-i\gamma)
\end{equation}
is true. Thus for any $f_{1},f_{2}$ from the domain $\mathcal{D}_{A} $ we get the expression
\begin{equation}\label{scat4a}\begin{array}{c}
\left(\mathfrak{S}(A_{L},A)f_{1},f_{2} \right)=(f_{1},f_{2})-	\underset{\varepsilon ,\eta\downarrow 0}{\lim}\:\iint\limits_{-\infty}^{\infty} \frac{[\mu+i(\gamma+\eta)][\lambda-i(\gamma-\varepsilon)]}{\lambda-\mu+i(\varepsilon-\eta)} \\ \times \left(\left[\frac{i\eta}{\mu-\lambda-i\varepsilon}\Gamma(\lambda+i\varepsilon)-\frac{i\varepsilon}{\lambda-\mu-i\eta}\Gamma(\mu+i\eta) \right]G(-i\gamma)^{*}dE_{\lambda}f_{1},  G(-i\gamma)^{*}dE_{\mu}f_{2}\right) 
\end{array}
\end{equation}
But in accordance with (\ref{G1}) for any $f\in\mathcal{H}$ we have
\begin{equation}\label{G2}
	G(-i\gamma)^{*}f=\sum_{j\in\mathbb{Z}}\left(f,\tilde{g}_{j} \right)_{\mathcal{H}}\mathbf{h}_{j}, \quad \tilde{g}_{j}=\sum_{j\in\mathbb{Z}}b_{jk}g_{k} .
\end{equation}
Substituting (\ref{G2})  into (\ref{scat4a}) and assuming that $f_{1},f_{2}\in E(\Delta)\mathcal{H}\cap\mathcal{D}_{A} $ yields
\begin{equation}\label{scat4b}\begin{array}{c}
\left(\mathfrak{S}(A_{L},A)f_{1},f_{2} \right)=\int_{\Delta}\left(\mathbf{f}_{1}(\lambda),\mathbf{f}_{1}(\lambda) \right)_{\mathcal{N}} d\lambda  -	\underset{\varepsilon,\eta\downarrow 0}{\lim}\iint_{\Delta\times\Delta}d\lambda d\mu \\ \times \frac{[\lambda-i(\gamma-\varepsilon)][\mu+i(\gamma-\eta)]}{2\pi\left[\lambda-\mu+i(\varepsilon-\eta)\right](\lambda-i\gamma)(\mu+i\gamma)}  \sum_{k,j\in \mathbb{Z}}\left[\frac{i\eta}{\mu-\lambda-i\varepsilon}\Gamma_{kj}(\lambda+i\varepsilon)-\frac{i\varepsilon}{\lambda-\mu-i\eta}\Gamma_{kj}(\mu+i\eta) \right] \\ \times \left(\mathbf{f}_{1}(\lambda),\hat{\mathbf{g}}_{j}(\lambda) \right)_{\mathcal{N}}\cdot  \left(\hat{\mathbf{g}}_{k}(\mu), \mathbf{f}_{2}(\mu)\right)_{\mathcal{N}}, \quad \Gamma_{kj}(z)=\left( \Gamma(z)\mathbf{h}_{j},\mathbf{h}_{k}\right) .
\end{array}	
\end{equation} 

To pass to the limit in (\ref{scat4b}), recall that in the Hilbert space $\mathbb{L}^{2}(\mathbf{R};\mathcal{M})$ of vector functions $\mathbf{f}(\lambda)$ on $\mathbf{R}$ with values in some Hilbert space $\mathcal{M}$ for any $\varepsilon >0$ operators $$\left(\pi_{\pm}^{(\varepsilon)}\mathbf{f} \right)(\lambda)=\pm \frac{1}{2\pi i}\int_{\mathbf{R}}\frac{1}{\mu-\lambda\mp i\varepsilon}\mathbf{f}(\mu)d\mu  $$ are contractions and strong limits $\pi_{\pm}$ of this operators as $\varepsilon\downarrow 0$  are orthogonal projectors  onto the Hardy subspace $\mathbf{H}^{2}_{+}(\mathcal{M})$ in the upper half-plane and its orthogonal complement $\mathbf{H}^{2}_{-}(\mathcal{M})$, respectively. In particular, for any $\mathbf{f}\in \mathbb{L}^{2}(\mathbf{R};\mathcal{M}$ we have
\begin{equation}\label{gen}\begin{array}{c}
\underset{\eta\downarrow 0}{\lim}\underset{\varepsilon\downarrow 0}{\lim}\frac{1}{2\pi i}\int_{\mathbf{R}}\left[ \frac{1}{\mu-\lambda- i\varepsilon}-\frac{1}{\mu-\lambda- i(\varepsilon-\eta)}\right] \mathbf{f}(\mu)d\mu \\	= \underset{\eta\downarrow 0}{\lim}\left[(\pi_{+}\mathbf{f})(\lambda) +\left(\pi_{-}^{(\eta)}\mathbf{f} \right)(\lambda)\right]=\mathbf{f}(\lambda). 
\end{array}
\end{equation}

Let us assume for a moment that the scalar functions $\left(\mathbf{f}_{s}(\lambda),\hat{\mathbf{g}}_{j}(\lambda) \right)_{\mathcal{N}}, \, s=1,2,$ in (\ref{scat4b}) satisfy the condition
\begin{equation}\label{help3}
	\int_{\Delta}\sum_{j\in\mathbb{Z}}\left|\left(\mathbf{f}_{s}(\lambda),\hat{\mathbf{g}}_{j}(\lambda) \right)_{\mathcal{N}} \right|^{2}d\lambda<\infty. 
\end{equation}
  In other words, let us assume that vector functions 
  $$\mathbf{f}_{s}(\lambda)=\sum_{j\in\mathbb{Z}}\left(\mathbf{f}_{s}(\lambda),\hat{\mathbf{g}}_{j}(\lambda) \right)_{\mathcal{N}}\mathbf{h}_{j}$$
  belong to the space $\mathbb{L}^{2}(\Delta;\mathcal{K})$.	 
  In this case, using relations (\ref{gen}), one can directly carry out the passage to the limit in (\ref{scat4b}) and, as a result, obtain the expression
  \begin{equation}
  \begin{array}{c}
  	\left(\mathfrak{S}(A_{L},A)f_{1},f_{2} \right)=\int_{\Delta}\left(\mathbf{f}_{1}(\lambda),\mathbf{f}_{2}(\lambda) \right)_{\mathcal{N}} d\lambda \\ -	2i\int_{\Delta} \sum_{k,j\in \mathbb{Z}} \left( \Gamma(\lambda+i0)\mathbf{h}_{j},\mathbf{h}_{k}\right) _{\mathcal{K}} \left(\mathbf{f}_{1}(\lambda),\hat{\mathbf{g}}_{j}(\lambda) \right)_{\mathcal{N}}\cdot  \left(\hat{\mathbf{g}}_{k}(\lambda), \mathbf{f}_{2}(\lambda)\right)_{\mathcal{N}}d\lambda  .
  \end{array}
\end{equation}
It remains to verify that the set of vector functions from $\mathbb{L}^{2}(\Delta;\mathcal{N})$ satisfying the condition (\ref{help3}) is dense in $\mathbb{L}^{2}(\Delta;\mathcal{N})$. To this end let us take the scalar function 
$$\psi(\lambda)=\sum_{j\in \mathbb{Z}}\left\|\tilde{\hat{g}}_{j}(\lambda) \right\|^{2}_{\mathcal{N}} .  $$  
According to the conditions of the theorem, $\psi(\lambda)$ is integrable since
$$\begin{array}{c}
\sum\limits_{j\in \mathbb{Z}}\int_{\Delta}\left\|\hat{\mathbf{g}}_{j}(\lambda) \right\|^{2}_{\mathcal{N}}d{\lambda} =\sum\limits_{j\in \mathbb{Z}}\int_{\Delta}\sum\limits_{k^{\prime},k\in\mathbb{Z}}\bar{b}_{jk^{\prime}}b_{jk}\left(\mathbf{g}_{k}(\lambda),\mathbf{g}_{k^{\prime}}(\lambda) \right)d\lambda \\ \leq M\sum\limits_{j,k^{\prime},k\in\mathbb{Z}}|b_{jk^{\prime}}||b_{jk}|=M\left(\sum_{j,k\in\mathbb{Z}}\left|b_{jk} \right| \right)^{2}<\infty. \end{array}
$$
For any $\mathbf{f}\in\mathbb{L}^{2}(\Delta;\mathcal{N})$ and any $\delta>0$ set $$ \mathbf{f}_{\delta}=\frac{1}{\sqrt{1+\delta\psi(\lambda)}}\mathbf{f}(\lambda).$$ Then, on the one hand, $$\int_{\Delta}\sum\limits_{j\in\mathbb{Z}}\left|\left(\mathbf{f}_{\delta}(\lambda),\hat{\mathbf{g}}_{j}(\lambda) \right)_{\mathcal{N}} \right|^{2}d\lambda\leq \int_{\Delta}\frac{\psi(\lambda)}{1+\delta\psi(\lambda)}\left\|\mathbf{f}(\lambda) \right\|^{2}_{\mathcal{N}}d\lambda\leq\frac{1}{\delta}\left\|\mathbf{f} \right\|^{2}  <\infty,$$ and on the other, it is obvious that $$\underset{\delta\downarrow 0}{\lim}\left\|\mathbf{f}-\mathbf{f}_{\delta} \right\| =0. $$	
\end{proof}
By our assumption the scattering operator $\mathfrak{S}(A_{1},A)$ is unitary. Therefore, the scattering matrix $S(A_{1}, A)(\lambda)$ is to be unitary almost everywhere on $\Delta$. Notice, that the fact that $$S(A_{1},A)(\lambda)^{*}S(A_{1},A)(\lambda)=I$$ follows directly from the relation
\begin{equation}\label{gramm}\begin{array}{c}
	\frac{1}{2i}\left(  \Gamma(\lambda+i0)^{-1}-\left[\Gamma(\lambda+i0)^{*}\right]^{-1}\right)_{j,k\in\mathbb{Z}} \\ =\frac{1}{2i}\left(  |L|^{-\frac{1}{2}}Q(\lambda+i0)|L|^{-\frac{1}{2}}-|L|^{-\frac{1}{2}}Q(\lambda+i0)^{*}|L|^{-\frac{1}{2}}\right)_{j,k\in\mathbb{Z}} \\ =\pi\left(\left(\hat{\mathbf{g}}_{j}(\lambda),\hat{\mathbf{g}}_{k}(\lambda) \right)_{\mathcal{N}} \right)_{j,k\in\mathbb{Z}}.
\end{array} 
\end{equation}
and general 

\begin{prop}[\cite{AdP1}]
\label{local2}
	Let $g_{1},...,g_{N}$ be a set of linearly independent vectors of the Hilbert space $ \mathcal{N}$, $\Upsilon=(\gamma_{nm})_{1}^{N}$ is the corresponding Gram-Schmidt matrix and $\Lambda$ is any Hermitian $N\times N$ - matrix such that the matrix $\Lambda+i\Upsilon$ is invertible. Then the matrix 
	\begin{equation}\label{local1}
		\Omega=I-2i \sum\limits_{n,m=1}^{N}\left([\Lambda+i\Upsilon)]^{-1}\right)_{nm}(.\, ,\, g_{n})g_{m}
	\end{equation}
	is unitary.
\end{prop}
The unitarity of the matrix $\Omega$ is verified by direct calculation, which is left to the reader. The same calculation can be repeated, and hence, \textit{ the statement of Proposition \ref{local2}  remains valid in the case of infinite sequences of linearly independent vectors $\{g_{j}\}_{j\in\mathbb{Z}} $ from $\mathcal{N}$  if 
\begin{itemize}
    \item {the closures $\hat{\Lambda}$ and $\hat{\Upsilon}$ of operators in $\mathbf{l}_{2} $ formally given as infinite \newline Hermitian matrices $\Lambda$ and $\Upsilon$ are selfadjoint operators;}
    \item{the domain of $\hat{\Upsilon}$ belongs to the domain of $\hat{\Lambda}$; }
    \item{the operator $\hat{\Lambda}+i\hat{\Upsilon}$ has bounded inverse.}
\end{itemize}
} 

\subsection{Clarification for perturbations of the Laplace operator by infinite sums of zero-range potentials}

Let $\mathbf{S}_{2}$ be the unit sphere in $\mathbf{R}_{3}$. The unitary mapping
\begin{equation}\label{four1}
\begin{array}{c}
\left(\mathfrak{F}f\right)(\lambda,\mathbf{n})=
\frac{\sqrt[4]{\lambda}}{\sqrt{2}(2\pi)^{\frac{3}{2}}} 
\int_{\mathbf{R}_{3}}f(\mathbf{x})e^{-i\sqrt{\lambda}(\mathbf{n}\cdot\mathbf{x})}d\mathbf{x}, 
\\ f\in \mathbf{L}_{2}\left(\mathbf{R}_{3}\right), \, \mathbf{n}\in \mathbf{S}_{2},
\end{array}
\end{equation} 
of $\mathbf{L}_{2}\left(\mathbf{R}_{3}\right)$ onto the space of vector function $\mathbf{L}_{2}\left(\mathbf{R}_{+};\mathbf{L}_{2}\left(\mathbf{S}_{2}\right)\right)$ with values from the Hilbert space $\mathbf{L}_{2}\left(\mathbf{S}_{2}\right)$ transforms the standard Laplace operator 
$A=-\Delta$ in $\mathbf{L}_{2}(\mathbf{R}_{3})$ into the selfadjoint operator of multiplication by the independent variable $\lambda$ in $\mathbf{L}_{2}\left(\mathbf{R}_{+};\mathbf{L}_{2}\left(\mathbf{S}_{2}\right)\right)$. Accordingly, the operator $A$ has the uniform Lebesgue spectrum of infinite multiplicity, filling the half-axis $[0, \infty)$ on the complex plane. At that, under the mapping $\mathfrak{F}$ the resolvent $R(z)$ of $A$ turns into the multiplication operator by the scalar function $(\lambda-z)^{-1}$ in $\mathbf{L}_{2}\left(\mathbf{S}_{2}\right)$, and the functions $g_{m}(z;\mathbf{x})$, appearing in Krein's formula (\ref{krein7}), are transformed into vector functions  
\begin{equation}\label{subop3D2a}
\mathbf{q}_{m}(\lambda,\mathbf{n})=\frac{\sqrt[4]{\lambda}}{4\pi}\frac{e^{-i\sqrt{\lambda}\left(\mathbf{n}\cdot\mathbf{x}_{m} \right)}}{\lambda-z}.
\end {equation}
Let $\left\lbrace \mathbf{x}_{j}\right\rbrace _{j\in\mathbb{Z}_{+}}$ be a sequence of different points of $\mathbf{R}_{3}$ that may have only a finite number of accumulation points in compact domains of $\mathbf{R}_{3}$ and an  $\left\lbrace \eta_{m}\right\rbrace _{m\in\mathbb{Z}_{+}}$ be the associated sequence of positive numbers $$ \eta_{m}=\underset{1\leq j\neq k\leq m}{\min}\left|\mathbf{x}_{j}-\mathbf{x}_{k}\right |, \, m=1,2,..., \, $$
We will also relate to $\left\lbrace \mathbf{x}_{j}\right\rbrace _{j\in\mathbb{Z}_{+}}$ the sequence positive-valued functions $\mu_{m}(\lambda)$ on the half-axe $\lambda>0$ which are the least eigenvalues for the sequence of Gramm-Schmidt matrices $\mathfrak{G}_{N}(\lambda )$ for the sets of functions
\begin{equation}\label{subop3D2ab}
\mathbf{e}_{j}(\lambda,\mathbf{n})=\frac{\sqrt[4]{\lambda}}{4\pi}e^{-i\sqrt{\lambda}\left(\mathbf{n}\cdot\mathbf{x}_{j} \right)}, \, j=1,...,N,.... 
\end {equation}
in the space $\mathbf{L}_{2}\left(\mathbf{S}_{2}\right)$, 
\begin{equation}\begin{array}{c}
\mathfrak{G}_{N}(\lambda )=\left(\left( \mathbf{e}_{m}(\lambda,\cdot),\mathbf{e}_{m^{\prime}}(\lambda,\cdot)\right)_{\mathbf{L}_{2}\left(S_{2}\right) } \right)_{m,\prime{m}=1}^{N} \\ 

= \left( \frac{\sin{\sqrt{\lambda}|\mathbf{x_{m}}-\mathbf{x}_{m^{\prime}}|}}{4\pi|\mathbf{x_{m}}-\mathbf{x}_{j^{\prime}}|}\left[1- \delta_{jj^{\prime}}\right]+\frac{\sqrt{\lambda}}{4\pi}\delta_{m^{\prime}} \right)_{m,m^{\prime}=1}^{N} 
\end{array}
\end{equation}
It is worth mentioning that
\begin{equation}\label{GS1}
    \mu_{N}(\lambda)=\|\mathfrak{G}_{N}(\lambda )^{-1}\|^{-1}.
\end{equation}

Note that  $\mu_{N}(\lambda)>0$ since\textit{ for any sequence of different points $\left\lbrace \mathbf{x}_{j}\right\rbrace _{j\in\mathbb{Z}_{+}}$ of $\mathbf{R}_{3}$ the functions $\mathbf{e}_{j}(\lambda,\mathbf{n})$ are linearly independent in the space of continuous functions $\mathbf{C}\left(\mathbf{S}_{2}\right)$ and as follows in the space $\mathbf{L}_{2}\left(\mathbf{S}_{2}\right)$.}

The last assertion is a simple generalization of the following, most likely known. 
\begin{prop}
 Let $\left\lbrace\mathbf{x}_{j}\right\rbrace_{j\in\mathbb{Z}_{+}}$ be a sequence of different points on the plain $\mathbf{R}_{2}$ with polar coordinates $\left(r_{j}\geq 0,\varphi\in[0,2\pi) \right)$, respectively. Then, for any $a>0$ the functions
 \begin{equation}\begin{array}{c}
w_{j}\zeta)=e^{iar_{j}\cos{\left(\theta-\varphi_{j}\right)}}=e^{iar_{j}\mathrm{Re}\zeta\cdot\bar{\zeta_{j}}}, \\\zeta=e^{i\theta}, \,\theta\in[0,2\pi], \quad \zeta_{j}=e^{i\varphi_{j}}, \end{array}
\end{equation}
 are linearly independent in the space of continuous function $\mathbf{C}([0,2\pi)])$.
 \end{prop}
\begin{proof}
Suppose that there are complex numbers $\alpha_{1},...,\alpha_{N}$ such that the linear combination
\begin{equation}\label{lind}
W(\zeta)=\sum\limits_{j=1}^{N}\alpha_{j}w_{j}(\zeta)\equiv 0, \, \,|\zeta|=1.
\end{equation}
Note that each function $w_{j}(\zeta)$ the boundary value on the unit circle of the function
\begin{equation}
    \hat{w}_{j}(z)=\exp{\left[a\left(\frac{z}{2\zeta_j}+ \frac{\zeta_j}{2z}\right)\right]},
\end{equation}
which is holomorphic in the open complex plane with pierced origin.
So, the uniqueness theorem for holomorphic functions and the identity (\ref{lind}) provide  
\begin{equation}\label{lind2}
   \hat{W}(z)=\sum\limits_{j=1}^{N}\alpha_{j}\hat{w}_{j}(z)\equiv 0, \, \, z\neq 0.
\end{equation}
On the other hand, without loss of generality, we can assume that point $\mathbf{x}_{N}$ lies on the positive semi-axis and that $r_{N}\geq r_{j},\, 1\leq j\leq N-1 $. Then, taking into account the identity (\ref{lind2}) and that the points $\mathbf{x}_{1},...,\mathbf{x}_{N}$ are different, we conclude that in the combination (1) the coefficient
\begin{equation*}
    \alpha_{N}=\underset{\tau\rightarrow \infty}{\lim} \exp\left({-\frac{ar_{N}}{2}}\tau\right)\hat{W}(-i\tau)=0.
\end{equation*}
Proceeding in the same way, one can verify that all the remaining coefficients in (\ref{lind}) are also zeros.
\end{proof}
Let $\left(w_{m}\right)_{m\in\mathbb{Z}_{+}}$  be a sequence of non-zero real numbers such that 
\begin{equation} \label{conddis}
K_{0}=\sum_{m \in \mathbb{Z}_{+}}\frac{1}{|w_{m}|} < \infty, \quad K_{1}=\sum_{m \in \mathbb{Z}_{+}}\frac{1}{\eta_{m}^{2}|w_{m}|} < \infty.
\end{equation}
Setting
\begin{equation*}
    M_{N}([a,b])=\underset{\lambda\in [a,b]}{\max}\|\mathfrak{G}_{N}(\lambda )^{-1}\|, \, 0<a<b<\infty, 
\end{equation*}
we will assume in addition, that the given sequence $\left(w_{m}\right)_{m\in\mathbb{Z}_{+}}$ is such  for any segment $[a,b]$ of the half-axis $(0,\infty)$ we have 
\begin{equation} \label{conddis1}
\underset{N\rightarrow\infty}{\lim} M_{N}([a,b])\sum\limits_{m = N+1}^{\infty} \frac{1}{\eta_{m}^{2}|w_{m}|} =0.
\end{equation}

By Theorem \ref{help1},  Krein’s formula (\ref{krein7}) for the points $\left\lbrace \mathbf{x}_{m}\right\rbrace$ and the diagonal matrix $L=\left(w_{m}\delta_{mn}\right)_{m,n\in \mathbb {Z}_{+}}$ generated by the sequence of non-zero real numbers $\left(w_{m}\right)_{m\in\mathbb{Z}_{+}}$ satisfying conditions (\ref{conddis}) gives the resolvent of the selfadjoint operator $A_{L}$ which is the close singular perturbation of the standard Laplace operator in the form of the infinite sum of disjoint zero-range potentials located at the given points, that is
\begin{equation}\begin{array}{c}
\mathcal{D}_{L}:=  \left\lbrace f: \, f=f_{0}+g, \, f_{0}\in {H}_{2}^{2}\left(\mathbf{R}_{3}\right), \, g\in \mathcal{N}\right. ,  \\

  \underset{\rho_{m}\rightarrow 0}{\lim}\left[\frac {\partial} {\partial\rho_{m}}\left(\rho_{m}f(\mathbf{x})\right)\right]+4\pi w_{m} \underset
	{\rho_{m}\rightarrow 0}{\lim}\,[\rho_{m}\,f(\mathbf{x})] =0, \\ \left. \rho_{n}=|\mathbf{x}-\mathbf{x}_{m}|, \quad m\in \mathbb{Z}_{+} \right\rbrace . \end{array} \end{equation}

Note that under Theorem \ref{help1}, if the conditions (\ref{conddis}) are met, the difference in the resolvents of operators ${A_{L}}$ and ${A}$ is a nuclear operator at their common regular points. Therefore, according to the Birman-Kato theorem (see \cite{Ya}), the absolutely continuous component of ${A_{L}}$ has the same spectrum as ${A}$.

To find the scattering matrix $S(A_{L}, A)(\lambda)$ for the given pair $A_{L}, A$, we will start with Theorem \ref{subopzer} and use the space of unilateral sequences $\mathbf{l}_{2}^{+}$ and the canonical basis $\mathbf{e}_{j} = \left(\delta_{jk}\right)_{j,k\in \mathbb{Z}_{+}}$ in $\mathbf{l}_{2}^{+}$ as the auxiliary Hilbert space $\mathcal{K}$ with its orthonormal basis $\left\lbrace \mathbf{h}_{j}\right\rbrace $.

\begin{thm}\label{subopf1}
If the conditions $(\ref{conddis})$, $(\ref{conddis1})$ holds, then the scattering matrix $S(\lambda)$ for the pair $A_{L}, A $ is continuous operator function with respect to the operator norm in $\mathbf{L}_{2}\left(\mathbf{S}_{2}\right)$, which acts as an integral operator  with the kernel  
\begin{equation}\label{subop3D1}
\begin{array}{c}
S(\mathbf{n},\mathbf{n}^{\prime};\lambda)=\delta\left(\mathbf{n}-\mathbf{n}^{\prime} \right) - \frac{\sqrt{\lambda}}{8i\pi^{2}}\sum\limits_{m,m^{\prime}=1}^{\infty}\Gamma_{mm^{\prime}}(\lambda+i0)%
\\ \times\frac{1}{\sqrt{\left|w_{m}\right|}} e^{-i\sqrt{\lambda}\mathbf{x}_{m}\cdot\mathbf{n}}\cdot \frac{1}{\sqrt{\left|w_{m^{\prime}}\right|}} e^{i\sqrt{\lambda}\mathbf{x}_{m^{\prime}}\cdot\mathbf{n}^{\prime}}, \quad  \mathbf{n},\mathbf{n}^{\prime} \in \mathbf{S}_{2}, \, \lambda>0; \\
\Gamma_{mm^{\prime}}(z)=\left(\left[J_{L}+\tilde{Q}(z)\right]^{-1}\right)_{mm^{\prime}}, \\ \left(J_{L}\right)_{mm^{\prime}}=\frac{w_{m}}{\left|w_{m} \right|}\cdot\delta_{mm^{\prime}}, 
\\ \tilde{Q}_{mm^{\prime}}(z)=\left(|L|^{-\frac{1}{2}}Q(z)L|^{-\frac{1}{2}}\right)_{mm^{\prime}} \\
=\left\{\begin{array}{c} \frac{i\sqrt{z}}{4\pi\left|w_{m} \right|}, \quad m=m^{\prime}, \\ \frac{1}{4\pi}\frac{e^{i\sqrt{z}|\mathbf{x}_{m}-\mathbf{x}_{m^{\prime}}|}}
{\sqrt{\left|w_{m}\right|}\cdot|\mathbf{x}_{m}-\mathbf{x}_{m^{\prime}}|\cdot\sqrt{\left|w_{m^{\prime}}\right|}}, \quad m\neq m^{\prime}, \end{array} \quad \mathrm{Im}z,\mathrm{Im}\sqrt{z}\geq 0. \right.  
\end{array}
\end{equation}    
\end{thm}
\begin{proof}
To prove the statement of the theorem, we must show that the Laplace operator $A$ and its close singular perturbation $A_{L}$ satisfy the conditions outlined in the theorem \ref{subopzer}. But due to the above, we need only to check for any finite interval $(a,b)$ on the semi-axis $\lambda>0$ the validity of the relations
\begin{equation}\label{ess+}
\underset{\lambda\in (a,b)}{\mathrm{esssup}} \left\|\Gamma(\lambda +i0)\right\|<\infty.	
\end{equation} 
when conditions (\ref{conddis}) are met.

To this end, first note that the elements of the matrix $\tilde{Q}(z)$ are bounded analytic functions in the upper half-plane with continuous boundary values on the real axis. Besides, 
\begin{equation}\label{bound11}\begin{array}{c}
  \left|\tilde{Q}_{mm^{\prime}}(z)\right|\leq\frac{|\sqrt{z}|}{4\pi\left|w_{m}\right|}\cdot\delta_{mm^{\prime}}  \\+\frac{1}{4\pi\eta_{\max\left\{m,m^{\prime}\right\}}\sqrt{\left|w_{m}\right|\cdot\left|w_{m^{\prime}}\right|}}\left( 1-\delta_{mm^{\prime}} \right), \, \mathrm{Im}z\geq 0.
\end{array}
\end{equation}
By (\ref{bound11}) the matrix $\tilde{D}(z)+\tilde{Q}(z)$ generates a bounded operator in $\mathbf{l}_{2}^{+}$ with the norm
\begin{equation}\label{bound12}\begin{array}{c}
 q_{L}(z):=\left\| \tilde{Q}(z)\right\|\leq p_{L}(z):=\underset{m}{\max}\frac{|\sqrt{z}|}{4\pi\left|w_{m}\right|}  \\
+\frac{1}{4\pi}\underset{m>1}{\max}\left(\left[\sum\limits_{m^{\prime}=1}^{m-1}\frac{1}{\eta^{2}_{m}%
\left|w_{m}\right|\left|w_{m^{\prime}}\right|}\right]^{\frac{1}{2}}+\left[\sum\limits_{m^{\prime}>m}^{\infty}\frac{1}{\left|w_{m}\right|\eta^{2}_{m^{\prime}}\left|w_{m^{\prime}}\right|}\right]^{\frac{1}{2}} \right)\\
\leq\underset{m>1}{\max} \frac{1}{4\pi \left | w\right|_{m}}\left(|\sqrt{z}|+\frac{1}{\eta^{2}_{m}}K_{0}+K_{1}\right), \, \mathrm{Im}z\geq 0.
\end{array}
\end{equation}
From inequality (\ref{bound12}) and conditions (\ref{conddis}), it follows that for sufficiently large values of numbers $\left |w_{m} \right|$, the norm $q_{L}(z), \, \mathrm{Im}z\geq 0,$ becomes less than unity and remains so at their further increase. At the same time, if inequality (\ref{bound12})  holds for some $z_{0}\neq 0, \, \mathrm{Im}z_{0}\geq 0$ with $p_{L}(z_{0})<1$, it guarantees this inequality for all $z$ of the half-disk $$\mathbb{D}(z_{0})=\{z:|z|\leq |z_{0}|, \, \mathrm{Im}\geq 0\}. $$
It is also apparent that the diagonal matrix $J_{L}$ with only either $+1$ or $-1$ on the main diagonal is unitary and, what is more, coincides with its inverse. Therefore for each $z\in \mathbb{D}(z_{0})$ and for any $\mathbf{h}\in\mathbf{l}_{2}^{+}$ with account of (\ref{bound12}) we conclude that
\begin{equation*}\begin{array}{c}
 \left\| \left[J_{L}+\tilde{Q}(z)\right]\mathbf{h}\right\|\geq \|J_{L}\mathbf{h}\|-\left\| \tilde{Q}(z)\mathbf{h}\right\|\\ \geq\|\mathbf{h}\|-q_{L}(z)\|\mathbf{h}\|\geq \left[1-p_{L}(z_{0})\right]
 \|\mathbf{h}\|>0.
 \end{array}
\end{equation*}
Hence, the $J_{L}+\tilde{Q}(z)$ as operator in $\mathbf{l}_{2}^{+}$ is invertible for $z\in \mathbb{D}_{0}$ and its inverse $\tilde{\Gamma}(z)$ is bounded with norm
\begin{equation}\label{normgam}
    \|\tilde{\Gamma}(z)\|\leq\frac{1}{1-q_{L}(z)}\leq\frac{1}{1-p_{L}(z_{0})}, \, z\in \mathbb{D}_{0}.
\end{equation}
As the operator function $J_{L}+\tilde{Q}(z)$ is uniformly bounded and strongly continuous on $\mathbb{D}_{0}$, then by virtue of (\ref{normgam}), so is its inverse $\tilde{\Gamma}(z)$. In particular, the operator function $\tilde{\Gamma}(\lambda +i0)$ is uniformly bounded and continuous with respect to the operator norm on the segment $[0,|z_{0}|]$.

To check the validity of the above statement for the case when the condition $p_{L}(\lambda+i0)<1, \, 0<a\leq\lambda\leq b<\infty $ may not be fulfilled, let us take enough large natural $N<\infty$ and represent $J_{L} +\tilde{Q}(z), \mathrm{Im}z\geq 0,$  in a block form 
\begin{equation}\label{bl}
    J_{L} +\tilde{Q}(z)=\left(\begin{array}{cc}W_{N}(z) & P_{N}(z)^{T} \\ P_{N}(z) & R_{N}(z) \end{array} \right),
\end{equation}
with the blocks 
\begin{equation}\label{dbl}\begin{array}{c}
\left(W_{N}(z)\right) _{m\,m^{\prime}}=\left(J_{L}+\tilde{Q}(z)\right)_{m\,m^{\prime}}, \, 1\leq m,m^{\prime}\leq N; \\
\left(R_{N}(z)\right)_{m\,m^{\prime}}=\left(J_{L}+\tilde{Q}(z)\right)_{m+N\, m^{\prime}+N}, \, 1\leq m,m^{\prime}; \\
\left(P_{N}(z)\right)_{m\,m^{\prime}}=\left(\tilde{Q}(z)\right)_{m+N\, m^{\prime}}, \, 1\leq m,\, 1\leq m^{\prime}\leq N.
\end{array}
\end{equation}
Note that diagonal blocks $W_{N}(z)$ and $R_{N}(z)$ would have appeared in Krein's formula (\ref{krein7})  as matrices $J_{L}+\tilde{Q}(z)$ for singular perturbations in the form of finite and infinite sums of zero range potentials, respectively, with locations and parameters $\left\lbrace \mathbf{x}_{j};w_{j}\right\rbrace _{1\leq j \leq N}$ and $\left\lbrace \tilde{\mathbf{x}}_{j}=\mathbf{x}_{j+N}; \tilde{w}_{j}=w_{j+N}\right\rbrace _{j\in\mathbb{Z}_{+}}$.
Remind that in the course of proof of Theorem \ref{help1} we have established that conditions (\ref{conddis}) guarantee that for $\mathrm{Im}z>0$ the matrices $J_{L}+\tilde{Q}(z)$ are invertible and once infinite, the corresponding operators in $\mathbf{l}_{2}^{+}$ generated them are boundedly invertible. As follows, this is also true for $W_{N}(z)$ and $R_{N}(z)$. 

Assuming further that (\ref{conddis}) hold and taking  a natural $N_{0}\geq 1$ so large that for given $b>0$ the inequality  
\begin{equation}\label{bl1a}
p_{N_{0}, L}(b)=\leq\underset{m>N_{0}}{\max} \frac{1}{4\pi \left | w\right|_{m}}\left(|\sqrt{b}|+\frac{1}{\eta^{2}_{m}}K_{0}+K_{1}\right)<1 
\end{equation} 
is true, by above we conclude that for any $N\geq N_{0}$ the operator function $\hat{R}_{N}(\lambda), \, 0<\lambda\leq b,$, generated in $\mathbf{l}_{2}^{+}$ by the boundary values by the matrix function $R_{N}(\lambda+i0)$ in (\ref{dbl}) is boundedly invertible and together with its inverse $\hat{R}_{N}(\lambda)^{-1}$ strongly continuous
on the interval $(0,b)$. With all that
\begin{equation}\label{bl2}
\left\| \hat{R}_{N}(\lambda)^{-1} \right\|\leq \frac{1}{1-p_{N_{0},L}(b)}, \, 0<\lambda\leq b. 
\end{equation}

At the same time, for any $N>0$ the matrix $\hat{W}_{N}(\lambda)=W_(\lambda+i0), \, \lambda>0,$ is invertible and, hence, the matrix function $\hat{W}_{N}(\lambda)^{-1}, \, \lambda >0,$ is continuous on the half-axe $\lambda >0$ since its imaginary part 
\begin{equation*}\begin{array}{c}
 \frac{1}{2i}\left[\hat{W}_{N}(\lambda)- \hat{W}_{N}(\lambda)^{*}\right]= \frac{1}{2i}\left(|L|^{-\frac{1}{2}}\left[\tilde{Q}_{N}(\lambda)- \tilde{Q}_{N}(\lambda)^{*}\right] |L|^{-\frac{1}{2}}\right)_{m,m^\prime =1}^{N}\\ =\left( |L|^{-\frac{1}{2}}\right)_{m,m^\prime =1}^{N}\mathfrak{G}_{N}(\lambda )\left( |L|^{-\frac{1}{2}}\right)_{m,m^\prime =1}^{N} >0,
\end{array}
\end{equation*}
as $\mathfrak{G}_{N}(\lambda )$ is the Gram-Schmidt matrix for the linearly independent in $\mathbf{L}_{2}\left(\mathbf{S}_{2}\right)$ set of functions $\mathbf{e}_{m}(\lambda,\mathbf{n}), \, m=1,...N,...$ (\ref{subop3D2ab}).

The invertibility of matrix $\hat{W}_{N}(\lambda)$ and the bounded invertibility of operator $\hat{R}_{N}(\lambda)$ generally do not give a guarantee of the bounded invertibility of operator $J_{L}+\hat{Q}(\lambda)$ generated by matrix $J_{L}+\tilde{Q}(\lambda+i0)$. However, in our case, they do so due to the additional conditions (\ref{conddis}) and (\ref{conddis1}). To be sure of this, taking advantage of the invertibility of block $\hat{R}_{N}(\lambda)=R_{N}(\lambda+i0)$ in expression (\ref{bl}), which, as we have seen, occurs by condition (\ref{conddis}) if $N\geq N_{0}$, and Schur-Frobenius factorization, let us represent the  right-hand side of (\ref{bl}) for $J_{L}+\hat{Q}(\lambda)$  as the product of three block matrices
\begin{equation}\label{ShFr} \begin{array}{c}
   \left(\begin{array}{cc}I & \hat{P}_{N}(\lambda)^{T}\hat{R}_{N}(\lambda)^{-1} \\ 0 & I \end{array} \right) \\ \times \left(\begin{array}{cc}\hat{W}_{N}(\lambda)-\hat{P}_{N}(\lambda)^{T}\hat{R}_{N}(\lambda)^{-1}\hat{P}_{N}(\lambda) & 0  \\ 0 & \hat{R}_{N}(\lambda)   \end{array} \right)\\  \times\left(\begin{array}{cc}I & 0 \\ \hat{R}_{N}(\lambda)^{-1}\hat{P}_{N}(\lambda) & I \end{array} \right).
\end{array}
\end{equation}
It can be seen that the extreme triangular multipliers in (\ref{ShFr}) are boundedly invertible. For  example,
\begin{equation*}
    \left(\begin{array}{cc}I & \hat{P}_{N}(\lambda)^{T}\hat{R}_{N}(\lambda)^{-1} \\ 0 & I \end{array} \right)^{-1}= \left(\begin{array}{cc}I & -\hat{P}_{N}(\lambda)^{T}\hat{R}_{N}(\lambda)^{-1} \\ 0 & I \end{array} \right). 
\end{equation*}

By bounded invertibility of $\hat{R}_{N}(\lambda)$ for $N\geq N_{0}$ the question of the bounded invertibility of the middle block diagonal factor reduces to the question of the invertibility of $N\times N$-matrix
\begin{equation}
\mathring{W}_{N}(\lambda) =\hat{W}_{N}(\lambda)-\hat{P}_{N}(\lambda)^{T}\hat{R}_{N}(\lambda)^{-1}\hat{P}_{N}(\lambda)
\end{equation}\label{ShFr4}
for $N\geq N_{0}$.

Turning to matrix $\mathring{W}_{N}(\lambda)$, recall that the imaginary part of the associated matrix $\hat{W}_{N}(\lambda)$coincides with the Gramm-Schmidt matrix $\mathfrak{G}_{N}(\lambda )$ for the first $N$ elements of a linearly independent sequence of functions (\ref{subop3D2ab}) in $\mathbf{L}_{2}\left(S_{2}\right)$. Therefore,
\begin{equation}\label{ShFr5}\begin{array}{c}
\mathrm{Im}\mathring{W}_{N}(\lambda)=\frac{1}{2i}\left[\mathring{W}_{N}(\lambda)-\mathring{W}_{N}(\lambda)^{*} \right]\\ = \mathfrak{G}_{N}(\lambda )-\mathrm{Im}\left[\hat{P}_{N}(\lambda)^{T}\hat{R}_{N}(\lambda)^{-1}\hat{P}_{N}(\lambda)\right] \\ \geq\left(\mu_{N}(\lambda)-\left\|\hat{P}_{N}(\lambda)^{T}\hat{R}_{N}(\lambda)^{-1}\hat{P}_{N}(\lambda) \right\|\right)\cdot I,
\end{array}
\end{equation}
where $I$ is the $N\times N$ unity matrix.

Note further that
\begin{equation}\label{ShFr6}\begin{array}{c}
 \left\|\hat{P}_{N}(\lambda) \right\|=  \left\|\hat{P}_{N}(\lambda)^{T} \right\|\leq \sqrt{\underset{1\leq m^{\prime}\leq N}{\max} \sum\limits_{m=N+1}^{\infty}\left|\tilde{Q}_{mm^{\prime}}(\lambda)\right|^{2}} \\ \leq \sqrt{\sum\limits_{m=N+1}^{\infty}\frac{1}{4\pi|w_{m}|\eta_{m}^{2}}}.
 \end{array}
\end{equation}
Due to (\ref{bl2}) and (\ref{ShFr6}),  we conclude that on every segment $0<a<b<\infty$ the inequality 
\begin{equation*}
    \left\|\mathfrak{G}_{N}(\lambda )^{-1}\hat{P}_{N}(\lambda)^{T}\hat{R}_{N}(\lambda)^{-1}\hat{P}_{N}(\lambda) \right\|\leq \frac{M_{N}([a,b])}{4\pi\left[1-p_{N_{0},L}(b)\right]}\sum\limits_{m = N+1}^{\infty} \frac{1}{\eta_{m}^{2}|w_{m}|}
\end{equation*}
can be satisfied  with a suitable choice of $N_{0}$ depending on $[a,b]$. But under (\ref{conddis1}), the right-hand side of this inequality tends to zero as $N\rightarrow \infty $. Therefore, we can conclude that due to condition (\ref{conddis1}), for any positive $\varepsilon <1$, the natural $N_{0}$ can be chosen such that for any $N\geq N_{0}$, along with condition $p_{N_{0}, L}(b)<1$, also the inequality 
\begin{equation}\label{ShFr7}
\mathrm{Im}\mathring{W}_{N}(\lambda) \geq\mu_{N}(\lambda)(1-\varepsilon)\cdot I,
\end{equation}
is satisfied. The latter guarantees the invertibility of the matrix $\mathring{W}_{N}(\lambda)$, and consequently the continuity of its elements on $[a,b]$.
\end{proof}


\begin{thebibliography}{1}
\bibitem{Ad} V. Adamyan, \textit {Singular perturbations of unbounded selfadjoint operators. Reverse approach}, Operator Theory: Advances and Applications, \textbf{276} ( 2020), 63 - 79
\bibitem{AdP} V. Adamyan and B. Pavlov, \textit{Null--range Potentials and M.G. Krein's formula for generalized resolvents}, Zap. Nauchn. Sem. Leningrad. Otd. Matemat. Inst. im. V.F. Steklova 
\textbf{149} (1986), 7--23 (J. of Soviet Math. \textbf{42}(1988), 1537--1550).
\bibitem{AdP1} V. Adamyan and B. Pavlov, \textit{Local Scattering Problem and a Solvable Model of Quantum Network}, Operator Theory: Advances and Applications, \textbf{198} ( 2009), 1-10
\bibitem{AGHH} S. Albeverio,  F. Gesztesy, R. Høegh-Krohn and  H. Holden,  \textit{Solvable models in quantum mechanics}, Texts and Monographs in Physics. 2nd edition, AMS-Chelsea Series, Amer. Math. Soc., 2005
\bibitem{AK} S. Albeverio and P. Kurasov, \textit{Singular Perturbations of Differential Operators},
Cambridge University Press, 2000.
\bibitem{BMN} J. Behrndt, M. Malamud and H. Neidhardt, \textit{Scattering matrices and Weyl functions},  Proc. Lond. Math. Soc., \textbf{97}, (2008), 568-598
\bibitem{BMN1} J. Behrndt, M. Malamud and H. Neidhardt, \textit{Scattering matrices and Dirichlet-to-Neumann maps}, Journal of Funct. Anal., \textbf{ 273}, (2017), 1970-2025
\bibitem{BF}  F.A. Berezin and L.D. Faddeev, |textit{Remarks on the Schr¨odinger equation with sin
gular potential}, Doklady Akad. Nauk SSSR, 137, (1961), 1011- 1014
\bibitem{DeOst} Yu. Demkov and V. Ostrovsky, \textit{Zero-range Potentials and their Applications
 in Atomic Physics}, Plenum, New York, (1988)
\bibitem{EKKST}  P. Exner, J.P. Keating, P. Kuchment, , T. Sunada, A.Teplyaev, \textit{ Analysis on
graphs and its applications}, Proceedings of Symposia in Pure Mathematics,
\textbf{77} (2008), Providence, RI. Am. Math. Soc. 
\bibitem{GHM} A. Grossmann, R. H\o{}egh-Krohn and M. Mebkhout, \textit{A class of explicitely soluble, local, many-center Hamiltonians for one-particle quantum mechanics in two and three dimensions}. I. J. Math. Physics, \textbf{21}(9) (1980), 2376-2385
\bibitem{Kat}T. Kato,\textit{Perturbation Theory for Linear Operators}, Springer-Verlag, Berlin, second edition, 1976 
\bibitem{Kr} M.G. Krein, \textit{Concerning the resolvents of an Hermitian operator with the deficiency-index $(m,m)$}, Doklady Acad. Sci. USSR, \textbf{52}, (1946), 651-654 
 \bibitem{Ya} D.R. Yafaev, \textit{Mathematical Scattering Theory}, Translations of Mathematical Monographs, American Math. Soc., Providence, RI, \textbf{105} (1992)     

\end{thebibliography}
\end{document}